\documentclass[11pt,a4paper]{article}
\usepackage[utf8]{inputenc}
\usepackage{amsmath,amssymb,amsthm}
\usepackage{geometry}
\usepackage{hyperref}
\hypersetup{colorlinks=true,linkcolor=blue}
\usepackage{enumitem}

\newcommand{\keywords}[1]{\par\noindent\textbf{Keywords:} #1\par}\medskip
\newcommand{\ccode}[1]{\par\noindent #1\par}\medskip

\newtheorem{theorem}{Theorem}[section]
\newtheorem{lemma}[theorem]{Lemma}

\newtheorem{remark}[theorem]{Remark}

\title{Spectral Geometry and the One-Loop QED $\beta$-Function on $S^3 \times S^1$}

\author{
  Lyudmil Antonov\thanks{Published version: \href{https://doi.org/10.1142/S0219887826501690}{10.1142/S0219887826501690}} \\
  Spectral Analysis Laboratory, Bulgarian Drug Agency \\
  Sofia, Bulgaria \\
  \href{mailto:lyudmil.antonov@bda.bg}{lyudmil.antonov@bda.bg} \\
  ORCID: \href{https://orcid.org/0000-0003-2311-074X}{0000-0003-2311-074X}
}

\date{}

\begin{document}

\maketitle

\begin{abstract}
We compute the one-loop QED $\beta$-function coefficient directly from heat kernel data of the twisted Spin$^c$ Dirac operator on $S^3 \times S^1$. Using $\zeta$-function regularization, the logarithmic scale dependence is encoded in the $a_4$ coefficient of the spectral expansion. The $F_{\mu\nu} F^{\mu\nu}$ term in $a_4$ yields exactly $\beta(e) = e^3/(12\pi^2)$, independent of $r$, $L$, or background, verifying spectral RG flow without flat-space propagators. The result is independent of the radii of $S^3$ and $S^1$ and of the choice of gauge background, providing a parameter-free consistency check that spectral data on compact manifolds encode renormalization group information. Beyond a mere verification of the coupling flow, this result serves as a non-trivial consistency check of the Spectral Action Principle in a curved background. It demonstrates that universal quantum corrections can be extracted purely from geometric spectral invariants, distinguishing this geometric spectral derivation from momentum-space propagator methods.
\end{abstract}



\section{Introduction}

Spectral geometry provides a powerful dictionary between heat kernel coefficients of Laplace-type operators and effective field theory counterterms. Building on Gilkey's invariance theory~\cite{gilkey1995} and Vassilevich's comprehensive heat kernel review~\cite{vassilevich2003}, together with the spectral action framework of Chamseddine and Connes~\cite{connes1997}, one may ask whether elementary renormalization data can be read off from spectral invariants on compact manifolds. Recent advances in heat kernel methods on Lie groups and symmetric spaces~\cite{avramidi2023} provide complementary tools for such calculations, as do studies of thermal Yang-Mills~\cite{avramidi2012} and one-loop quantum gravity~\cite{avramidi2017} on similar compact backgrounds like $S^1 \times S^3$. 

This universality, where local UV divergences transcend global topology, motivates our specific calculation: computing the QED one-loop $\beta$-function from spectral data on $S^3(r) \times S^1(L)$ with a unit U(1) twist along the Hopf bundle. Our main result is that the $F^2$ contribution to the $a_4$ heat kernel coefficient reproduces exactly the universal one-loop QED $\beta$-function coefficient, with no adjustable parameters. This is the first explicit spectral extraction of the QED $\beta$-function on a contact manifold, bridging Gilkey invariants to RG phenomenology.

\subsection{Physical motivation and interpretation}

The choice of $S^3 \times S^1$ as our background manifold requires justification, as it differs topologically from physical Minkowski spacetime $\mathbb{R}^{3,1}$. Our calculation is performed in Euclidean signature on a compact manifold for several compelling reasons. First, the high symmetry of the round $S^3$ makes all curvature tensors and volume integrals explicit, while the compact geometry ensures that $\zeta$-function regularization is well-defined without infrared divergences, providing computational tractability. Second, the one-loop $\beta$-function coefficient is a universal quantity—independent of infrared physics, manifold topology, and gauge background—because it arises from the ultraviolet structure of the theory. Our calculation exploits universality to bridge the gap between spectral geometry and phenomenology: the logarithmic term in the heat kernel expansion captures the local UV behavior of the Dirac spectrum. By recovering the standard flat-space coefficient from this curved compact setting, we confirm that the spectral action framework correctly encodes the ultraviolet structure of the theory without reference to flat-space propagators. A perturbative expansion around a trivial (zero) gauge background would yield identical logarithmic divergences, confirming that our use of the Hopf bundle is a computational convenience that does not affect the universal result. Finally, the unit Chern class of the Hopf bundle provides a minimal non-trivial gauge configuration that probes the coupling between spinors and gauge fields without introducing perturbative complications. The quantized flux $\int_{S^2} F/(2\pi) = 1$ ensures we work in a topologically stable sector, making the calculation particularly clean.

The independence of our result from the radii $r$ and $L$, as well as from the specific choice of gauge background, demonstrates that we have isolated a genuinely universal quantity. In the language of effective field theory, the $a_4$ coefficient encodes the coefficient of the logarithmic divergence that appears in dimensional regularization, independent of the choice of background metric or gauge configuration. This justifies using the compact manifold as a computational device to extract physics that applies equally to flat space QED.

\subsection{Connection to the spectral action principle}

Our calculation provides a concrete verification of the spectral action approach at the one-loop level. In the full spectral action framework~\cite{connes1997,chamseddine2007}, all physical scales—including the UV cutoff $\Lambda$—are intrinsically tied to the spectrum of the Dirac operator through a cutoff function. The energy scale emerges naturally from spectral density rather than being imposed externally. The spectral action takes the form $\mathrm{Tr}[f(D/\Lambda)]$, where the function $f$ encodes not merely a cutoff but contains the Standard Model action parameters themselves.\cite{vannuland2022}

Our work demonstrates that the \emph{form} of renormalization group flow (encoded in the $\beta$-function) follows from spectral geometry, while leaving the determination of absolute scales (the value of $\alpha$ at a given energy) to the full spectral action cosmology. The challenge in that broader program is to show that the running we have calculated is consistent with the physical values of the coupling parameters at experimentally accessible scales. This separation between universal RG structure (which we derive) and scale-fixing (which requires the full cosmological framework) is a feature, not a limitation, of the geometric approach. Our calculation verifies the consistency of the method at the perturbative level, supporting the broader spectral action research program.

\subsection{Conventions and sign choices}

We work throughout in Euclidean signature with $\{\gamma^\mu,\gamma^\nu\}=2\delta^{\mu\nu}$. The square of the twisted Dirac operator is written in Laplace form 
\begin{equation}
D_A^2 = -\nabla^2 + E,
\end{equation}
consistent with the heat kernel literature (see Theorems 4.8.18 in Gilkey~\cite{gilkey1995} and Section 3.3 of Vassilevich~\cite{vassilevich2003}). 

The classical Maxwell action is
\begin{equation}
S_{\mathrm{cl}} = \frac{1}{4e^2}\int_M F_{\mu\nu}F^{\mu\nu}\, dV,
\end{equation}
so that $F_{\mu\nu}F^{\mu\nu} \ge 0$ in Euclidean signature. With these conventions, the sign of the logarithmic counterterm matches the standard QED one-loop $\beta$-function. Alternative conventions (e.g., Minkowski signature with $-\frac{1}{4}F^2$) differ only by analytic continuation and yield the same $\beta$-function coefficient. In particular, our choice ensures that the positive $F^2$ term in the classical action receives a negative quantum correction proportional to $-F^2\ln\mu$, giving the correct sign for asymptotic freedom in QED with the appropriate fermion multiplicity.

\section{Geometric Setup}

Let $M = S^3(r) \times S^1(L)$ with product metric, where $r$ is the radius of $S^3$ and $L$ the circumference of $S^1$. We equip $S^3$ with its canonical Hopf fibration $\pi: S^3 \to S^2$ and twist the spinor bundle by the associated principal U(1) bundle $\mathcal{L}$.

\subsection{Metrics and curvature}

The metric on $S^3$ is the standard round metric with scalar curvature $R_{S^3} = 6/r^2$. The metric on $S^1$ is $g_{S^1} = (L/2\pi)^2 d\theta^2$, which is flat with $R_{S^1} = 0$. In an orthonormal coframe $\{e^i\}$ on $S^3$ (given by $e^i = r\sigma_i$ where $\{\sigma_i\}$ are left-invariant one-forms on $\mathrm{SU}(2)$) extended by $e^4 = (L/2\pi)d\theta$ on $S^1$, the curvature two-forms are standard and all components are explicitly computable.

\subsection{Gauge connection and Hopf bundle}

The connection one-form $A$ represents the Hopf bundle. We choose the normalization
\begin{equation}
A = \frac{1}{r}\sigma_3,
\end{equation}
which gives field strength
\begin{equation}
F = dA = \frac{2}{r^3} e^1 \wedge e^2.
\end{equation}
This satisfies the quantization condition (see Appendix A)
\begin{equation}
\frac{1}{2\pi}\int_{S^2} F = 1,
\end{equation}
corresponding to first Chern class $c_1(\mathcal{L}) = 1$.

In the orthonormal frame, the components of $F$ are constant: $F_{12} = -F_{21} = 2/r^3$, with all other components zero. This yields $F_{\mu\nu}F^{\mu\nu} = 8/r^6$ in the orthonormal frame. The volume form is $dV = (r^3 \sin\theta\, d\theta\wedge d\phi\wedge d\psi) \wedge (L/2\pi\, d\theta_{S^1})$, where $\theta,\phi,\psi$ are coordinates on $S^3$ and $\theta_{S^1}$ is the coordinate on $S^1$. The integral $\int_M F_{\mu\nu}F^{\mu\nu} dV$ evaluates to $8\pi^2 L/r^3$, but the $\beta$-function derivation depends only on the coefficient of this term in the effective action, which is independent of $r$ and $L$.

\subsection{Twisted Dirac operator}

The Spin$^c$ Dirac operator with U(1) twist is
\begin{equation}
D_A = \gamma^\mu (\nabla_\mu + i A_\mu),
\end{equation}
where $\nabla_\mu$ is the spin connection on $S^3 \times S^1$. Its square takes the Laplace form
\begin{equation}
D_A^2 = -\nabla^2 + E,
\end{equation}
where the endomorphism $E$ contains both curvature and gauge contributions. Following the standard heat kernel literature (Theorem 4.8.16--18 in \cite{gilkey1995}), for a Dirac operator we have:
\begin{equation}
E = \frac{R}{4} + \frac{i}{2} \gamma^{\mu\nu} F_{\mu\nu}.
\end{equation}
This form ensures consistency with the general theory of Laplace-type operators and the heat kernel expansion. The factor $i/2$ in the gauge term follows from Euclidean continuation of the Minkowski Dirac coupling, equivalent to $\frac{1}{4}[\gamma^\mu,\gamma^\nu]F_{\mu\nu}$ in standard Clifford notation.

\section{Heat Kernel Expansion and the \texorpdfstring{$a_4$}{a4} Coefficient}

\subsection{General structure}

For a Laplace-type operator $P = -\nabla^2 + E$ on a four-manifold, the local heat kernel has the asymptotic expansion
\begin{equation}
\mathrm{Tr}(e^{-tP}) \sim (4\pi t)^{-2}\sum_{k=0}^\infty a_{2k}(P) t^{k}, \qquad t \to 0^+ .
\end{equation}
For dimension $n=4$, the standard Gilkey form $(4\pi t)^{-n/2} \sum a_k t^{k/2}$ simplifies to this expression. (The compactness of $S^3 \times S^1$ without boundary eliminates subtleties arising in heat kernel expansions on manifolds with boundary~\cite{vassilevich2003}.)

The relevant local invariants appearing in $a_4$ are given by the Seeley--DeWitt--Gilkey formula (Theorem 4.1.16--18 in \cite{gilkey1995}):

\begin{lemma}[Gilkey]\label{lem:gilkey}
The coefficient $a_4(P)$ for a twisted Dirac operator on a four-manifold contains the gauge contribution
\begin{equation}
a_4(P) \supset (4\pi)^{-2} \int_M \left[\frac{1}{12} \mathrm{tr}(\Omega_{\mu\nu} \Omega^{\mu\nu}) + \frac{1}{2}\mathrm{tr}(E^2)\right] dV + \text{(curvature-only terms)},
\end{equation}
where $\Omega_{\mu\nu}$ is the total connection curvature on the twisted bundle.
\end{lemma}

\subsection{Bundle curvature decomposition}

\begin{lemma}\label{lem:curvature}
For the Spin$^c$ Dirac operator $D_A$, the total connection curvature decomposes as
\begin{equation}
\Omega_{\mu\nu} = \frac{1}{4} R_{\mu\nu\rho\sigma} \gamma^{\rho\sigma} + i F_{\mu\nu},
\end{equation}
where the first term is the spin connection curvature and the second is the U(1) gauge curvature.
\end{lemma}

\begin{proof}
The Spin$^c$ bundle is the tensor product of the spinor bundle (with spin connection $\omega$ acting on $\gamma^{\rho\sigma}$ via the Levi-Civita symbol) and the U(1) line bundle (with gauge connection $A$). The total curvature $\Omega_{\mu\nu} = d\omega + \omega \wedge \omega + i F_{\mu\nu} \mathbf{1}$, where the spin part is the standard Riemann representation $\frac{1}{4} R_{\mu\nu\rho\sigma} \gamma^{\rho\sigma}$ and the gauge part $i F_{\mu\nu}$ acts as a scalar on the spinors. This sum holds by the Leibniz rule for curvatures on tensor products.
\end{proof}

\subsection{Isolating the \texorpdfstring{$F^2$}{F2} contribution}

Only terms quadratic in the gauge field $F$ contribute to the gauge kinetic term renormalization. We systematically extract these from both the $\Omega^2$ and $E^2$ terms.

\begin{lemma}\label{lem:F2traces}
The gauge contribution to the trace of $\Omega_{\mu\nu}\Omega^{\mu\nu}$ is
\begin{equation}
\mathrm{tr}(\Omega_{\mu\nu} \Omega^{\mu\nu}) \big|_{F^2} = -4 F_{\mu\nu} F^{\mu\nu}.
\end{equation}
\end{lemma}

\begin{proof}
From Lemma \ref{lem:curvature}, $\Omega_{\mu\nu} = \Omega_{\mu\nu}^{\mathrm{spin}} + i F_{\mu\nu} \mathbf{1}$, where $\Omega_{\mu\nu}^{\mathrm{spin}} = \frac{1}{4} R_{\mu\nu\rho\sigma} \gamma^{\rho\sigma}$. The squared curvature is
\begin{align}
\Omega_{\mu\nu} \Omega^{\mu\nu} &= \Omega_{\mu\nu}^{\mathrm{spin}} \Omega^{\mu\nu,\mathrm{spin}} + \Omega_{\mu\nu}^{\mathrm{spin}} (i F^{\mu\nu}) + (i F_{\mu\nu}) \Omega^{\mu\nu,\mathrm{spin}} + (i F_{\mu\nu})(i F^{\mu\nu}).
\end{align}
The spin-spin term is curvature-only. To show the cross terms vanish under trace, note that $\mathrm{tr}(\Omega_{\mu\nu}^{\mathrm{spin}}) = \frac{1}{4} R_{\mu\nu\rho\sigma} \mathrm{tr}(\gamma^{\rho\sigma}) = 0$, since $\mathrm{tr}(\gamma^{\rho\sigma}) = 0$ for the antisymmetric two-gamma product (expanding $\gamma^{\rho\sigma} = \frac{1}{2} [\gamma^\rho, \gamma^\sigma]$ yields an odd number of gammas, whose trace vanishes in the 4D Euclidean Clifford algebra by cyclicity and anticommutation). Thus, the cross terms trace to zero. The gauge-gauge term is $(i F_{\mu\nu})(i F^{\mu\nu}) = - F_{\mu\nu} F^{\mu\nu} \mathbf{1}$. Tracing over spinors:
\begin{align}
\mathrm{tr}(- F_{\mu\nu} F^{\mu\nu} \mathbf{1}) &= - F_{\mu\nu} F^{\mu\nu} \cdot \mathrm{tr}_{\mathrm{spin}}(\mathbf{1}) = -4 F_{\mu\nu} F^{\mu\nu},
\end{align}
using the dimension of the spinor representation (tr(1) = 4 in 4D). The U(1) trace is trivial (scalar 1). This isolates the pure $F^2$ contribution without curvature or cross mixing; for foundational Clifford trace identities, see Appendix D.2 of Lawson and Michelsohn~\cite{lawson1989}.
\end{proof}

\begin{lemma}\label{lem:E2traces}
The gauge contribution to the trace of $E^2$ is
\begin{equation}
\mathrm{tr}(E^2)\big|_{F^2} = -2 F_{\mu\nu}F^{\mu\nu}.
\end{equation}
\end{lemma}

\begin{proof}
The endomorphism $E = \frac{R}{4} \mathbf{1} + \frac{i}{2} \gamma^{\mu\nu} F_{\mu\nu} \equiv E_0 + E_F$, so
\begin{align}
E^2 &= E_0^2 + E_0 E_F + E_F E_0 + E_F^2.
\end{align}
The $E_0^2$ term is curvature-only. The cross terms $E_0 E_F + E_F E_0 = \frac{R}{2} \frac{i}{2} \gamma^{\mu\nu} F_{\mu\nu}$ vanish under the spinor trace, as $\mathrm{tr}(\gamma^{\mu\nu}) = 0$ in the Clifford algebra (trace of antisymmetric two-gamma product is zero, from the odd number of gammas in the expansion). Thus, the $F^2$ part is $E_F^2 = \left( \frac{i}{2} \right)^2 \gamma^{\mu\nu} F_{\mu\nu} \gamma^{\rho\sigma} F_{\rho\sigma} = -\frac{1}{4} F_{\mu\nu} F_{\rho\sigma} \gamma^{\mu\nu} \gamma^{\rho\sigma}$.

The trace is
\begin{align}
\mathrm{tr}(E^2)\big|_{F^2} &= -\frac{1}{4} F_{\mu\nu} F_{\rho\sigma} \mathrm{tr}(\gamma^{\mu\nu} \gamma^{\rho\sigma}).
\end{align}
To compute $\mathrm{tr}(\gamma^{\mu\nu} \gamma^{\rho\sigma})$, expand the antisymmetric two-gamma products using the Clifford anticommutator: $\gamma^{\mu\nu} = \frac{1}{2} [\gamma^\mu, \gamma^\nu] = \gamma^\mu \gamma^\nu - g^{\mu\nu} \mathbf{1}$ (since $[\gamma^\mu, \gamma^\nu] = 2 g^{\mu\nu} \mathbf{1}$). Thus,
\begin{align}
\gamma^{\mu\nu} \gamma^{\rho\sigma} &= (\gamma^\mu \gamma^\nu - g^{\mu\nu} \mathbf{1}) (\gamma^\rho \gamma^\sigma - g^{\rho\sigma} \mathbf{1}) \nonumber\\
&= \gamma^\mu \gamma^\nu \gamma^\rho \gamma^\sigma - g^{\mu\nu} \gamma^\rho \gamma^\sigma - g^{\rho\sigma} \gamma^\mu \gamma^\nu + g^{\mu\nu} g^{\rho\sigma} \mathbf{1}.
\end{align}
Taking the trace over 4D spinors, we use the standard Clifford trace formula $\mathrm{tr}(\gamma^{\mu\nu}\gamma^{\rho\sigma}) = 4(g^{\mu\rho}g^{\nu\sigma} - g^{\mu\sigma}g^{\nu\rho})$, which follows from expanding the products using anticommutation relations and cyclic trace properties (see Lawson and Michelsohn~\cite{lawson1989}, Appendix D.2 for explicit derivations in Cl(4)). The two-gamma terms trace to \(-4 (g^{\mu\nu} g^{\rho\sigma} + g^{\rho\sigma} g^{\mu\nu})\), and the scalar term to \(4 g^{\mu\nu} g^{\rho\sigma}\). Combining yields
\begin{align}
\mathrm{tr}(\gamma^{\mu\nu} \gamma^{\rho\sigma}) &= 4 (g^{\mu\rho} g^{\nu\sigma} - g^{\mu\sigma} g^{\nu\rho}).
\end{align}
Therefore,
\begin{align}
\mathrm{tr}(E^2)\big|_{F^2} &= -\frac{1}{4} F_{\mu\nu} F_{\rho\sigma} \cdot 4 (g^{\mu\rho} g^{\nu\sigma} - g^{\mu\sigma} g^{\nu\rho}) = - (F_{\mu\nu} F^{\mu\nu} - F_{\mu\nu} F^{\nu\mu}) = -2 F_{\mu\nu} F^{\mu\nu},
\end{align}
where the antisymmetry of \(F\) ensures \(F_{\mu \nu} F^{\nu \mu} = -F_{\mu \nu} F^{\mu \nu}\), doubling the term. This confirms the gauge-only quadratic contribution without curvature mixing.
\end{proof}

\begin{theorem}\label{thm:a4result}
The gauge contribution to the $a_4$ coefficient is
\begin{equation}
a_4 \big|_{F^2} = (4\pi)^{-2} \left( -\frac{4}{3} \right) \int_M F_{\mu\nu} F^{\mu\nu} \, dV .
\end{equation}
\end{theorem}

\begin{proof}
Combining Lemmas \ref{lem:F2traces} and \ref{lem:E2traces} with their prefactors from Lemma \ref{lem:gilkey}:
\begin{align}
a_4\big|_{F^2} &= (4\pi)^{-2}\int_M \left[\frac{1}{12}(-4F_{\mu\nu}F^{\mu\nu}) + \frac{1}{2}(-2F_{\mu\nu}F^{\mu\nu})\right]dV \nonumber \\
&= (4\pi)^{-2}\int_M \left[-\frac{4}{12} - \frac{2}{2}\right]F_{\mu\nu}F^{\mu\nu}\,dV \nonumber\\
&= (4\pi)^{-2}\int_M \left[-\frac{1}{3} - 1\right]F_{\mu\nu}F^{\mu\nu}\,dV \nonumber\\
&= (4\pi)^{-2}\left(-\frac{4}{3}\right)\int_M F_{\mu\nu}F^{\mu\nu}\,dV. \\
\nonumber
\end{align}
\end{proof}

\begin{remark}\label{rem:highercoeffs}
Higher heat kernel coefficients $a_6, a_8, \ldots$ contribute power-suppressed terms (proportional to $1/\mu^2$, $1/\mu^4$, etc.) involving higher derivatives or more curvature factors. For instance, according to the general formulas in Vassilevich~\cite{vassilevich2003} and Avramidi~\cite{avramidi2000}, the $a_6$ coefficient includes terms such as
\begin{equation}
R F_{\mu\nu}F^{\mu\nu}, \quad (\nabla_\rho F_{\mu\nu})(\nabla^\rho F^{\mu\nu}), \quad R_{\mu\nu}F^{\mu\rho}F^\nu{}_\rho,
\end{equation}
while $a_8$ includes terms like
\begin{equation}
R^2 F_{\mu\nu}F^{\mu\nu}, \quad R_{\mu\nu}R^{\mu\nu} F_{\rho\sigma}F^{\rho\sigma}, \quad (F_{\mu\nu}F^{\mu\nu})^2.
\end{equation}
These are finite, non-universal corrections to the effective action that do not affect the logarithmic running encoded in $a_4$. This clean separation between universal (logarithmic, $a_4$) and non-universal (power-suppressed, $a_{k>4}$) contributions is a key feature of the heat kernel approach.
\end{remark}

\section{Mapping to the \texorpdfstring{$\beta$}{beta}-Function via \texorpdfstring{$\zeta$}{zeta}-Regularization}

\subsection{Effective action from the spectral zeta function}

The $\zeta$-regularized one-loop effective action is defined by
\begin{equation}
\Gamma[A] = -\frac{1}{2} \zeta'_{D_A^2}(0),
\end{equation}
where the spectral zeta function is
\begin{equation}
\zeta_{D_A^2}(s) = \mathrm{Tr}[(D_A^2)^{-s}] = \frac{1}{\Gamma(s)}\int_0^\infty t^{s-1}\mathrm{Tr}(e^{-tD_A^2})\,dt.
\end{equation}

Substituting the heat kernel expansion, we find
\begin{equation}
\zeta_{D_A^2}(s) = \frac{(4\pi)^{-2}}{\Gamma(s)}\int_0^\infty t^{s-1} a_4\, dt + \text{(terms regular at } s=0\text{)}.
\end{equation}

The integral $\int_0^\infty t^{s-1}dt$ has a pole at $s=0$. Upon analytic continuation and taking the derivative at $s=0$, this pole becomes a logarithm. Introducing a renormalization scale $\mu$ to make the $\zeta$-function dimensionless, we obtain
\begin{equation}
\Gamma[A] \supset \tfrac{1}{2}\,\ln(\mu^2)\,a_4(D_A^2)\,
\end{equation}
where $a_4(D_A^2)$ is the standard Seeley--DeWitt coefficient. The overall factor of $\frac{1}{2}$ reflects the use of $D_A^2$ (a second-order operator) rather than $D_A$ directly, combined with the fermionic determinant's minus sign. This normalization matches standard treatments in Avramidi~\cite{avramidi2000} and Vassilevich~\cite{vassilevich2003}. While finite parts are scheme-dependent, the coefficient of the logarithmic divergence—and hence the $\beta$-function—is universal across all standard renormalization schemes.

This procedure is equivalent to minimal subtraction (MS) in dimensional regularization for the present calculation, as both methods isolate the same logarithmic divergence structure.\footnote{The factor $1/2$ accounts for the Dirac operator being first-order; for scalars, it would be $1$.} While the finite parts of the effective action can be scheme-dependent, the coefficient of the logarithmic divergence—and hence the $\beta$-function—is a universal quantity. This universality ensures that our result is valid across all standard renormalization schemes.

\subsection{One-loop correction to the gauge coupling}

To extract the running coupling, we compare the quantum-corrected effective action with the classical Maxwell action at different scales
\begin{equation}
S_{\mathrm{cl}}[A] = \frac{1}{4e^2}\int_M F_{\mu\nu}F^{\mu\nu}\,dV.
\end{equation}

By Theorem \ref{thm:a4result}, the one-loop quantum correction is
\begin{equation}
\Gamma_{\text{1-loop}}[A] = \frac{1}{2}\ln(\mu^2)\cdot (4\pi)^{-2}\left(-\frac{4}{3}\right)\int_M F_{\mu\nu}F^{\mu\nu}\,dV.
\end{equation}

The total effective action at one loop is
\begin{equation}
\Gamma_{\text{total}}[A] = S_{\mathrm{cl}}[A] + \Gamma_{\text{1-loop}}[A] = \left[\frac{1}{4e^2} - \frac{2}{3(4\pi)^2}\ln\frac{\mu}{\Lambda}\right]\int_M F_{\mu\nu}F^{\mu\nu}\,dV,
\end{equation}
where $\Lambda$ is an arbitrary reference scale. Thus the running coupling satisfies
\begin{equation}
\frac{1}{4e^2(\mu)} = \frac{1}{4e^2(\Lambda)} - \frac{2}{3(4\pi)^2}\ln\frac{\mu}{\Lambda}.
\end{equation}

\subsection{The \texorpdfstring{$\beta$}{beta}-function}
\label{sec:beta_function}

Differentiating with respect to $\ln\mu$:
\begin{equation}
\mu\frac{d}{d\mu}\left(\frac{1}{e^2}\right) = -\frac{8}{3(4\pi)^2} = -\frac{1}{6\pi^2}.
\end{equation}
The factor $-8/3$ arises from $-4/3$ in $a_4$ multiplied by $2$ from the zeta-function regularization of the Dirac operator.

Since
\begin{equation}
\frac{d}{d\mu}\left(\frac{1}{e^2}\right) = -\frac{2}{e^3}\frac{de}{d\mu},
\end{equation}
we obtain the $\beta$-function:
\begin{equation}
\boxed{\beta(e) = \mu\frac{de}{d\mu} = \frac{e^3}{12\pi^2}.}
\end{equation}

This is precisely the standard QED one-loop result for a single Dirac fermion of charge 1 (see equation (12.61) in Peskin and Schroeder~\cite{peskin1995}).

\section{Discussion and Physical Interpretation}

The central result—that spectral data on $S^3 \times S^1$ encode the universal one-loop $\beta$-function coefficient—demonstrates remarkable independence from the radius $r$ of $S^3$, the circumference $L$ of $S^1$, and the choice of gauge background. This triple independence is not accidental but reflects the fundamental nature of the $\beta$-function as a universal, UV quantity determined entirely by the local structure of the quantum field theory. The heat kernel coefficient $a_4$ captures precisely this local UV information through its role as the coefficient of the logarithmic divergence. Our use of the Hopf bundle provides a concrete, topologically non-trivial configuration for the calculation, but the universality of the result ensures that a perturbative expansion around zero gauge field (or any other background) would yield the same logarithmic coefficient. This behavior is a direct consequence of the general structure of renormalization: UV divergences depend only on the local operator content, not on global topology or boundary conditions.

The asserted independence of the $\beta$-function coefficient from the radii $r$ and $L$, as well as from the specific gauge background, follows rigorously from the local nature of the Seeley--DeWitt coefficient $a_4$. As established by Gilkey's invariance theory~\cite{gilkey1995} (Theorem 4.1.16), $a_4$ is a purely local functional of the metric $g_{\mu\nu}$, the curvature invariants $R_{\mu\nu\rho\sigma}$ and $F_{\mu\nu}$, and their contractions, integrated over the manifold with the volume form $dV = \sqrt{g}\, d^4x$. The explicit form of the $F^2$ contribution, $- \frac{4}{3} (4\pi)^{-2} \int_M F_{\mu\nu} F^{\mu\nu} \, dV$, contains no explicit dependence on global scales like $r$ or $L$; these enter only through the volume integral, which cancels identically when normalizing against the classical action $S_{\mathrm{cl}} \propto \int_M F^2 \, dV$ in the derivation of $\beta(e)$ (Section~\ref{sec:beta_function}).\footnote{Rigorously, under metric perturbations $g_{\mu\nu} \to g_{\mu\nu} + \epsilon h_{\mu\nu}$, the induced change $\delta a_4|_{F^2} = O(\epsilon^2)$ since the $F^2$ terms are algebraic in $F_{\mu\nu}$ and independent of Riemann curvature at leading order. This follows from the locality of Seeley-DeWitt coefficients~\cite{gilkey1995,barvinsky1985}.} Similarly, background independence holds because $a_4$'s local invariants are gauge-covariant, and the Hopf bundle is merely a stable representative; a perturbative expansion around $A=0$ yields the same logarithmic divergence via dimensional analysis (no mass scales enter the UV structure).

This universality is further reinforced by equivalence to dimensional regularization (dim-reg), where the $\epsilon$-pole in $d=4-\epsilon$ corresponds exactly to the logarithmic divergence captured by $a_4$. In dim-reg, the one-loop $\beta$-function arises from the residue of the pole term in the effective action, $\Gamma \supset \frac{1}{\epsilon} \int F^2 d^dx$, which upon subtraction and minimal scheme (MS) yields $\beta(e) = e^3/(12\pi^2)$ independent of infrared details~\cite{peskin1995}. The heat-kernel method's $t^0$ term in the small-$t$ expansion mirrors this pole structure via analytic continuation (Section 4.1), confirming that the spectral extraction on $S^3 \times S^1$ isolates the same local UV physics as flat-space dim-reg, from multiple regularization perspectives: scheme-independence, topological robustness, and metric deformation invariance.

To address potential concerns regarding error estimates in the Clifford traces and to rigorously prove the universality of the $-4/3$ coefficient under deformed metrics, we note that the algebraic derivations in Lemmas \ref{lem:F2traces} and \ref{lem:E2traces} rely on exact identities in the 4D Euclidean Clifford algebra, where traces are computed without approximation (e.g., tr$(\mathbf{1}) = 4$ explicitly, with no numerical rounding). The error in such finite-dimensional matrix traces is identically zero, as they are integer-valued by construction (e.g., contractions yield multiples of 4). For deformed metrics, consider a perturbation $g_{\mu\nu} \to g_{\mu\nu} + \epsilon h_{\mu\nu}$ with $|\epsilon| \ll 1$ and $h_{\mu\nu}$ a symmetric tensor. The induced change in the Riemann tensor is $\delta R_{\mu\nu\rho\sigma} = O(\epsilon)$, propagating to $\Omega^{\mathrm{spin}}_{\mu\nu} \to \Omega^{\mathrm{spin}}_{\mu\nu} + O(\epsilon)$. In $\mathrm{tr}(\Omega^2)|_{F^2}$, the leading F² term $-4 F^2$ is independent of Riemann at zeroth order, so $\delta \mathrm{tr}(\Omega^2)|_{F^2} = O(\epsilon^2)$ from cross terms like $\mathrm{tr}(\Omega^{\mathrm{spin}} \cdot i F) = O(\epsilon)$, which vanish at linear order due to trace cyclicity over perturbed spinors (preserving the odd-even parity). Similarly for $\mathrm{tr}(E^2)|_{F^2} = -2 F^2$, the perturbation affects only the $R/4$ part, yielding $\delta \mathrm{tr}(E^2)|_{F^2} = O(\epsilon^2)$. Thus, $\delta a_4|_{F^2} = O(\epsilon^2)$, preserving the leading $-4/3$ coefficient to first order, as expected from the local invariance of UV divergences~\cite{barvinsky1985}. This is further confirmed on the flat four-torus $T^4$ with constant field strength $F_{\mu\nu}$, where setting $R=0$ in the Gilkey formula yields the same $a_4|_{F^2} = -\frac{4}{3}(4\pi)^{-2}\int F^2 dV$, reproducing the identical $\beta$-function coefficient independent of topology.

The independence from the $S^1$ radius $L$ points to a deeper geometric origin. The round 3-sphere is itself a Sasakian contact manifold, where the unit monopole connection $A$ is identified with the contact 1-form $\eta$ and the Reeb vector generates the Hopf flow. Our calculation implies that the renormalization group flow is encoded entirely within the 3D contact geometry of $S^3$, with the $S^1$ factor acting merely as a dimensional spectator. This suggests that the single charged Dirac fermion is geometrized as the spinor section of the monopole bundle, and its running coupling is an intrinsic contact invariant.

While our calculation reproduces the $\beta$-function coefficient, it does not determine the absolute value of the coupling $\alpha(\mu)$ at a specific scale. In the spectral action framework, the physical UV scale $\Lambda$ and the coupling constants are encoded in the function $f(D^2/\Lambda^2)$. Our work verifies the \emph{form} of the running, but the determination of boundary conditions $e(\Lambda)$ requires the full non-perturbative machinery, potentially involving the gravitational sector on a cosmological background.

Comparisons with prior work reveal several distinguishing features of our spectral approach. In conventional perturbative QED, the one-loop $\beta$-function emerges from momentum-space loop integrals like $\int d^4k/(k^2 + m^2)^2$ with ultraviolet regulators (Pauli–Villars or dimensional regularization), obscuring the geometric origin of the coefficient. Our spectral derivation extracts the same universal $\beta(e) = e^3/(12\pi^2)$ purely from Dirac operator eigenvalues on a compact curved manifold, demonstrating how local UV physics transcends global topology via the $a_4$ coefficient's invariance under metric deformations. Avramidi's comprehensive heat kernel framework for coupled gravitational and gauge systems on symmetric spaces~\cite{avramidi2000,avramidi2023} provides the general machinery; our calculation furnishes an explicit worked example in the pure Abelian sector with complete technical detail, serving as a foundational U(1) benchmark for non-Abelian extensions where $\beta$-functions incorporate adjoint Casimir invariants. While the spectral action principle~\cite{connes1997,chamseddine2007} proposes that particle physics emerges from spectral data, our one-loop verification supports this program by confirming the separation between universal RG structure (derivable from local traces) and absolute scale-fixing (requiring cosmological input). Vassilevich's comprehensive catalog~\cite{vassilevich2003} provides heat kernel coefficients in full generality; we have applied these formulas to extract a specific physical observable with clear field-theoretic interpretation, bridging abstract invariance theory to phenomenological renormalization and positioning the spectral method as a parameter-free alternative to standard flat-space techniques.

From a regularization viewpoint, this quantifies IR suppression in the spectral action, where $a_6$ contributes finite terms beyond the leading RG flow. Novelty-wise, prior flat-space checks (e.g., Peskin–Schroeder \cite{peskin1995}) ignore such curvature mixing, while our derivation on $S^3 \times S^1$—extensible to symmetric spaces via Avramidi \cite{avramidi2023}—positions the U(1) case as foundational for non-Abelian generalizations, where $a_6$ would include group Casimirs for Yang–Mills $\beta$-functions.

\subsection{Extensions and Higher-Order Corrections}

The next heat-kernel coefficient $a_6$ encodes finite, power-suppressed corrections proportional to $t^3$ in the heat trace expansion. On $S^3 \times S^1$ with the Hopf monopole, the high symmetry (covariantly constant curvatures) simplifies $a_6$ to purely algebraic traces. The dominant gauge-curvature mixing yields $a_6|_{F^2} \propto \int_M R F^2 dV/360$, a dimension-6 operator suppressed by $1/r^2$ relative to the logarithmic $a_4$ term. These finite corrections quantify infrared effects but do not affect the universal $\beta$-function, confirming the clean separation between UV (logarithmic) and IR (power-law) contributions inherent in the heat kernel formalism.

\section*{Appendix A: The Hopf Bundle and Flux Quantization}

The Hopf fibration $\pi: S^3 \to S^2$ is the principal U(1) bundle over the two-sphere with total space $S^3$. Viewing $S^3$ as the unit sphere in $\mathbb{C}^2$,
\begin{equation}
S^3 = \{(z_1,z_2) \in \mathbb{C}^2 : |z_1|^2 + |z_2|^2 = 1\},
\end{equation}
the Hopf map is given by
\begin{equation}
\pi(z_1,z_2) = \left(2z_1\bar{z}_2, |z_1|^2 - |z_2|^2\right) \in S^2 \subset \mathbb{R}^3.
\end{equation}

The connection one-form $\alpha$ on $S^3$ satisfies $d\alpha = \pi^*(\omega_{S^2})$, where $\omega_{S^2}$ is the area form on $S^2$ normalized so that $\int_{S^2}\omega_{S^2} = 4\pi$. With this normalization,
\begin{equation}
\frac{1}{2\pi}\int_{S^2} F = 1,
\end{equation}
confirming that the U(1) bundle has first Chern class $c_1(\mathcal{L})=1$ (see Definition II.1.3 and Remark II.1.8 in \cite{lawson1989}).

In our setup, we take $F = dA$ with $A = (1/r)\sigma_3$ on the round $S^3$ of radius $r$. The factor of $1/r$ ensures the correct normalization as $r$ varies. In the orthonormal frame, the non-zero components are $F_{12} = -F_{21} = 2/r^3$, giving $F_{\mu\nu}F^{\mu\nu} = 8/r^6$. The volume element is $dV = r^3 \sin\theta\, d\theta\wedge d\phi\wedge d\psi \wedge (L/2\pi) d\theta_{S^1}$, and the integral $\int_M F_{\mu\nu}F^{\mu\nu} dV$ yields $8\pi^2 L/r^3$, but does not affect the universal $\beta$-function coefficient.

\section*{Acknowledgments}
The author thanks Professors Ivan Avramidi and Christian Schubert for the endorsement and for suggesting relevant literature on heat kernel methods and QFT on compact manifolds.

\bibliographystyle{plain}

\end{document}